\documentclass[journal]{IEEEtran}
\usepackage{pifont}\usepackage{amssymb} \usepackage{mathrsfs}
\usepackage{mathbbold}
\usepackage{multirow}
\usepackage{booktabs}
\usepackage{cite}
\usepackage{bm}
\usepackage{mathbbold}
\usepackage{flushend}
\usepackage{amsmath}

   \usepackage{graphicx}
   \graphicspath{{../pdf/}{../jpeg/}}
\usepackage{amsmath}
\usepackage{amsfonts}
\makeatother
\newtheorem{theorem}{Theorem}
\newtheorem{lemma}{Lemma}
\newtheorem{proof}{Proof}[section]

\begin{document}
\title{Analysis on the Empirical Spectral Distribution of Large Sample Covariance Matrix and Applications for Large Antenna Array Processing}
%
%
% author names and ieee memberships
% note positions of commas and nonbreaking spaces ( ~ ) latex will not break
% a separate \thanks must be used for each paragraph as latex2e's \thanks

\author{Guanping~Lu, ~\IEEEmembership{Member,~IEEE},\
Jinsong~Wu, ~\IEEEmembership{Senior Member,~IEEE},\
%~, ~\IEEEmembership{Member,~IEEE},\
Robert~C.~Qiu,~\IEEEmembership{Fellow,~IEEE}\

\thanks{This work was supported by the National Natural Science Foundation of China under Grant no. 61401274,
G. P. Lu is with the School of Management, Fudan University, Shanghai, he is also with Shanghai Gold Exchange, Shanghai, China (e-mail: lugp@sge.com.cn).
Jinsong Wu is with Department of Electrical Engineering, Universidad de Chile, Chile.
Robert Qiu is with the Academy of Information Technology and Electronic Engineering, Shanghai Jiaotong University.

The authors would like to thank Xin Wang and three anonymous
referees for their comments and suggestions which have been helpful in improving the
quality of this paper.}}% <-this % stops a space

\maketitle

\begin{abstract}
%\boldmath

This paper addresses the asymptotic behavior of a particular type of information-plus-noise-type matrices, where the column and row number of the matrices are large and of the same order, while signals are diverged and time delays of the channel are fixed. We prove that the empirical spectral distribution (ESD) of the large dimension sample covariance matrix and a well-studied spiked central Wishart matrix converge to the same distribution. As an application, an asymptotic power function is presented for the general likelihood ratio statistics for testing the presence of signal in large array signal processing.
 %We find that the with the  which shows that the results coincide with our analysis.

%this paper describes a frame structure aimed at solving the uplink
%channel in tds-ofdm system
\end{abstract}
% ieeetran.cls defaults to using nonbold math in the abstract.
% this preserves the distinction between vectors and scalars. however,
% then you can use latex's standard command \boldmath at the very start
% of the abstract to achieve this. many ieee journals frown on math

% note that keywords are not normally used for peerreview papers.

% for peer review papers, you can put extra information on the cover
% \ifclassoptionpeerreview
% \begin{center} \bfseries edics category: 3-bbnd \end{center}
% \fi
%
% for peerreview papers, this ieeetran command inserts a page break and
% creates the second title. it will be ignored for other modes.
\IEEEpeerreviewmaketitle

%\begin{document}
\section{Introduction}
%\begin{CJK*}{GBK}{kai}
The signal detection of information-plus-noise-type matrices is fundamental to modern wireless communication networks. Due to the growing scale of network and limited time resource, the available sample sizes cannot be quite larger than the dimensions. For this issue, the traditional covariance matrix theory is actually not applicable, which needs much larger sample sizes than signal dimensions\cite{ttony}. Thus, random matrix theory (RMT) has been used in resolving the high dimension estimation problems in signal processing\cite{Debbah}\cite{lin12}.

%, such as in \cite{couillet} and \cite{loubaton}, the hypothesis tests are constructed using joint fluctuations of the
%eigenvalues and eigenvectors of a large dimensional sample covariance
%matrix; in \cite{alex}, the spiked correlation model is proposed to do the detection.
The work \cite{lin12}\cite{Tulino} studied the user detection and signal detection algorithms using large dimension signal-plus-noise matrices. It was shown that the random matrix theory on large dimension signal-plus-noise matrices could guide the large antenna array signal processing very well. On the other hand, the work \cite{silver} researched the empirical spectral distribution (ESD) convergence behavior of three large dimension sample covariance matrices (SCM) categories. Given a $P\times P$ Hermitian matrix $\mathbf{R}$, for any real $x$, the ESD, $F(x)$, is defined by
\begin{equation}
F(x) \leftarrow \frac{1}{P}\#\{\lambda_j:\lambda_j<x\}
\end{equation}
where $\#E$ denotes the cardinality of the set $E$, $\lambda_j$ is the $j$-th eigenvalue. \cite{wang14} and \cite{yao18} investigated the general likelihood ratio test (GLRT) for linear spectral statistics of the eigenvalues of high-dimensional SCM from Gaussian populations. There is a gap between the traditional signal detection model as \cite{Tulino} and state-of-art random matrices theory \cite{yao18}, which, hopefully, can be connected by our research as a bridge. %This paper addresses the connection between the LSS of Dozier type matrices and the CLT for l.s.s. of population covariance matrix. %The same matrix category is researched by Bai and S According to former research[4, 4.3.2], when the sensor dimensions and sampling number are comparable, the noise covariance matrix is a centered high-dimensional Wishart matrix. Signal detection is researched in diversified high dimensional scenarios. In \cite{loubaton}\cite{loubaton2}, the ESD is high dimensional signal detection problem in the presence of signal, where the ESD character is assumed and unproven. In statistic field, \cite{Dozier} proves that, the ESD of a particular sample covariance matrix is weak convergence to a particular distribution. But the relationship between this ESD result and signal detection problem is not built.

In this work, a simple kind of information-plus-noise-type matrices with isotropic noise and limited signal dimension is studied. It is proved that the ESD of the received SCM converges in distribution to the ESD of central spike Wishart matrix\cite{baik}. Based on this feature, GLRT in \cite{wang14} has been used in hypothesis test of the signal detection for massive MIMO systems. In this paper, the GLRT tests were applied in time domain to support the hybrid transmission scheme such as filter bank multicarrier (FBMC) or filtered multitone (FMT) modulation \cite{guanpinglu}. Except the scenario in this paper, the results could be used widely from MIMO detection\cite{Debbah} to multiuser detection\cite{Tulino}. 

By proposed hypothesis test method, we evaluate the detection performance of large antenna array using the same number of antennas but the different number of samples. it is found that the simulation results highly agree with those from the theoretical analysis.

In this paper, $P$ is the antenna number and $N$ is the sample number which are $N$ and $n$ in \cite{silver} and Lemma 1, $s(n)$ is the n-th element of sequence $s$, $\mathbf{s}$ is the row vector, $[\mathbf{x}_1;\mathbf{x}_2;\cdots;\mathbf{x}_n]$ means constructing a matrix by $\mathbf{x}_l$ as rows, $diag(v)$ returns a square diagonal matrix with the elements of vector v on the main diagonal, $^*$ is the Hermitian transpose operator, $\mathbf{X}$ is the received signal matrix, $\mathbf{X}(i,l)$ is the $i,l$-th element of $\mathbf{X}$ which is $\mathbf{X}_{i,l}$ in \cite{silver} and Lemma 1, $N(0,1)$ is the Normal distribution, $\mathcal{CN}(0, 1)$ is the complex Normal distribution, $\mathbb{C}$ and $\mathbb{R}$ are complex number field and real number field, respectively.

%In this paper, signals received by different antenna are put together to form a information-plus noise matrix.  We prove that the sampled correlation matrix of this matrix is an elliptic ensemble which is widely researched in big data\cite{rqiu}, and the spectral distribution of this elliptic ensemble is equal to a central Wishart matrix. Furthermore,

%The remainder of this paper is organized as follows. In section II, we provide the signal model under hypothesis $\mathcal{H}_0$, and we prove that the correlation matrix coincides with the Wishart matrix assumption. In section III, we recall the LRT construction algorithm and establish the expectation and variance of statistics using the central limit theorem. In section IV, the simulation results are presented.

\section{System Model}

In signal detection, SCM analysis is used to explore the fundamental limits of communication\cite{Tulino}. Consider that there is a single transmitter which sends the signal $s(n)$ to $P$ receivers (or antenna elements) in an $L$-length channel as in \cite{Debbah}. For single antenna element of $i$, the received signal $\mathbf{x}_i$ is the summation of the $L$ tap delayed transmitted signal vector $\mathbf{s}$ with propagation as $h_{i,l}$, where $l = 0,...,L-1$. Accordingly, $N$-length points are sampled from each antenna, which is written as
\begin{equation}\label{ana}
\mathbf{x}_i=\sum_{l=0}^{L-1}h_{i,l}\mathbf{s}_l+\mathbf{w}_i
\end{equation}
where $h$ is the concatenate factor of transmitter amplifier and channel, $ h_{i,l}\sim \mathcal{CN}(0, \sigma^2)$, $s(n)$ is the transmit signal and follows the Binomial distribution as $P(s(n)=K)=(1/2)^{1-(K+1)/2}(1/2)^{(1+K)/2}, K \in\{-1,1\}$ or Normal distribution as $N(0, 1)$. $\mathbf{w}_i$ is the noise vector at the receiver containing independent identically distributed (i.i.d.) complex entries and unit variance, $\sigma^2$ is the signal power on each taps. Suppose the receiver truncates length-$N$ signal for detection. The length of channel delay is not larger than $L$ samples. The $l-lag$ of the signal vector has elements as $\mathbf{s}_l=\{s(-l),\cdots,s(N-l+1)\}$.

The task is to detect whether the signal is presented by processing the SCM of the receiver. It turns to be a hypothesis testing problem, where $\mathcal{H}_0$ means that the signal does not exist, and $\mathcal{H}_1$ means that the signal exists. The received signal samples under the hypothesis test are given, respectively, as
\begin{equation}\label{disc0}
\begin{split}
&\mathcal{H}_0: \mathbf{x}_i=\mathbf{w}_i;\\
&\mathcal{H}_1: \mathbf{x}_i=\sum_{l=0}^{L-1}h_{i,l}\mathbf{s}_l+\mathbf{w}_i.
\end{split}
\end{equation}

%which adapts for both block transmission and generalized serial transmission signals.
Furthermore, it is assumed that the noise and channel propagation are uncorrelated. The complex Gaussian noise has the property as
\begin{equation}\label{disc}
\mathbf{E}(\mathbf{w}(i))=0;
\mathbf{E}(\mathbf{w}_i\mathbf{w}_{j(\neq i)}^*)=0;
\mathbf{E}(\mathbf{w}_i^*\mathbf{w}_i)=\mathbf{I}_{N};
\end{equation}
The channel propagation has the same property.

In large scale antenna systems, the signal detection could be based on the joint operation of the sampled signal from each antenna element, such as in \cite{Debbah}. The sample vector $\mathbf{x}_n$ from each antenna elements is formed by $N$ samples, and it could construct a matrix as
\begin{equation}\label{dh}
\mathbf{X}=[\mathbf{x}_1;\mathbf{x}_2;\cdots;\mathbf{x}_P]
\end{equation}
where $\mathbf{X}$ has a dimension of $P\times N$ with element as $\mathbf{X}(i,j)$ and $P\rightarrow\infty$,$N\rightarrow\infty$,$P/N\rightarrow c$. The maximum channel length $L=O(1)$ is a fixed value regardless the value of $N$. The SCM of \eqref{dh} is
\begin{equation}\label{suma}
\mathbf{R_x}=\mathbf{X}\mathbf{X}^{*}/N.
\end{equation}

Note $F(x)$ as the ESD of the SCM. Under hypothesis $H_0$, the SCM is not a good approximate of the covariance matrix \eqref{disc}. As stated by Machenko-Pastur theorem, almost surely
\begin{equation}
F'(x)=\frac{1}{2\pi xc}\sqrt{(b-x)(x-a)},a< x< b,
\end{equation}
and $a=(1-\sqrt{c})^2$, b=$(1+\sqrt{c})^2$ when $0< c\leq 1$. When $c>1$, there is an additional Dirac measure at $x=0$ of mass $1-\frac{1}{c}$.
The $F'(x)$ is also named as Machenko-Pastur (M-P) Law\cite{pastur99}. Its shape is demonstrated in Fig. 1, which has parameter as $P=256, c=1/2$ and $1/8$.

\begin{figure}\centering
	\includegraphics[scale=0.5]{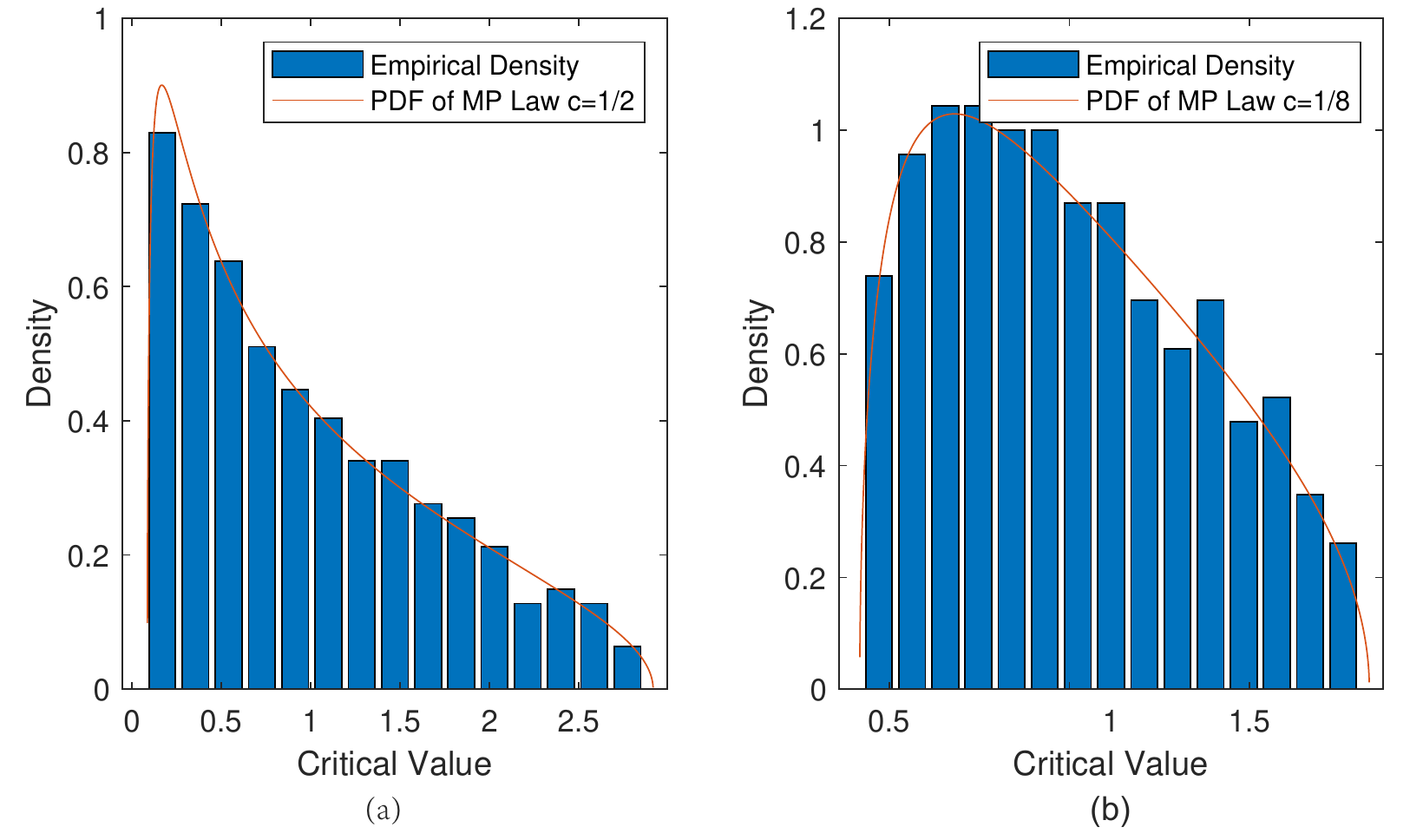}\\
	\caption{MP law under $P=256$, (a) $c=1/2$ and (b) $c=1/8$.}
\end{figure}

\section{Sample covariance Matrix Analysis}

Under hypothesis $\mathcal{H}_1$, the similar ESD analysis of the SCM has been proposed, such as in \cite{Debbah}. However, they did not provide the ESD when $L>1$, and the method is quite explicit\cite{yao18}. Here, we provide a theoretical result that the ESD of $H_1$ is the same as that of a central Wishart matrix, which has been studied extensively.

The received signal matrix $\mathbf{X}$ is further decomposed as $\mathbf{X}=\mathbf{S}_r+\mathbf{W}$, where $\mathbf{S}_r$ is the unknown signal matrix transmitted in $\mathbf{X}$, so
\begin{equation}\label{suma1}
\mathbf{R_x}=(\mathbf{S}_r\mathbf{S}_r^*+\mathbf{S}_r\mathbf{W}^*+\mathbf{W}\mathbf{S}_r^*+\mathbf{W}\mathbf{W}^*)/N,
\end{equation}
where $\mathbf{S}_r=\mathbf{H}\mathbf{S}_L$, $\mathbf{H}$ is a $P\times L$ matrix with elements as $h_{i,l}$ in \eqref{ana}, $\mathbf{S}_L=[\mathbf{s}_0 ; \cdots ;\mathbf{s}_{L-1}]$ has a size of $L\times N$, and the last part of $\mathbf{R_x}$ is the correlation of white Gaussian noise.

The SCM in \eqref{suma} is analyzed as
\begin{equation}\label{th1}
\begin{split}
&\mathbf{R}_x=\mathbf{T}+\mathbf{W}\mathbf{W}^*/N,
\end{split}
\end{equation}
where
\begin{equation}\label{th1LN}
\begin{split}
&\mathbf{T}=\mathbf{H}\mathbf{S}_L\mathbf{S}_L^*\mathbf{H}^*/+(\mathbf{H}\mathbf{S}_L\mathbf{W}^*+\mathbf{W}\mathbf{S}_L^*\mathbf{H}^*)/N\\
&=\mathbf{H}\mathbf{H}^*+(\mathbf{H}\mathbf{Q}+\mathbf{Q}^*\mathbf{H}^*)/\sqrt N),
\end{split}
\end{equation}
where $\mathbf{Q}=\mathbf{S}_L\mathbf{W}^*/\sqrt N$ is a random matrix with dimension $L\times P$. In \eqref{th1LN}, $\lim\limits_{n\rightarrow \infty}\mathbf{S}_L\mathbf{S}_L^*/N=\mathbf{I}_{N}$, and $\lim\limits_{n\rightarrow \infty}{\mathbf{H}\mathbf{S}_L\mathbf{S}_L^*\mathbf{H}^*/N}= \mathbf{H}\mathbf{H}^*$. 

For the convenience of portraying the matrix entires, it is assumed that $s(i)$ follows Binary distribution as $P(s(i)=k)=1/2,k\in\{-1,1\}$. Random variable $rv=\mathbf{S}_L(l,j) \mathbf{W}(j,p)$ follows $\mathcal{CN}(0,1)$, because if $s(n)$ and $w(n)$ follows Binary distribution and standard complex Normal distribution, 
\begin{equation}
\begin{split}
P\{s(n)w(n)>T\}&= \frac{P\{w(n)>T\}}{2}-\frac{P\{w(n)<-T\}}{2}\\
&=P\{w(n)>T\}
\end{split}
\end{equation}
 Entries $\mathbf{Q}(l,p)= \sum_{j=0}^{N-1}(\mathbf{S}_L(l,j) \mathbf{W}(j,p))/\sqrt{N}$ follows $\mathcal{CN}(0,1)$. Entries in $\mathbf{Q}^*$ have the same character.

If the signal follows Normal distribution, 
$\mathbf{Q}(l,p)$ converges to Normal distribution according to Law of Large Number. Simulations show that the Normal distributed signal has the same character as Binary distributed ones but with a slower convergence speed.

\begin{lemma}[perturbation on SCM\cite{silver}, Theorem 1.1]

For $N=1,2,\cdots$, $\mathbf{X}_N = \frac{1}{\sqrt{N}}(X_{i,j}^N)$, $N\times n$, $X_{i,j}^N \in \mathbb{C}$, identically distributed, independent across $n,i,j$ for each $N$, $E|X_{11}^1 - EX_{11}^1 |^2 = 1$, $n/N \rightarrow c > 0$ as $N \rightarrow \infty$.\\
a)	$T_N=diag(t_1^n,\cdots,t_n^n)$, $t_i^n \in \mathbb{R}$ and the distribution function of $\{t_1^n,\cdots,t_n^n\}$ converges almost surely in distribution to a probability distribution function (PDF) $H$ as $N \rightarrow \infty$.\\
b)  $A_N$ is Hermitian $N \times N$ for which the ESD $F(A_N)$ converges vaguely to $\mathcal{A}$ almost surely, $\mathcal{A}$ being a nonrandom distribution function.\\
c)	$X_N$, $T_N$, and $A_N$ are independent.\\
Let $B_N = A_N + X_NT_NX_N^*$. Then, almost surely, $F(B_N)$, the ESD of $B_N$, converges vaguely, as $N\rightarrow \infty$, to a (non-random) distribution function $\hat{F}$, whose Stieltjes transform $m(z) (z \in \mathbb{C}^+)$ satisfies
\begin{equation}\label{perturb}
m=m_\mathcal{A}(z-c\int \frac{t}{1+tm}dH(t))
\end{equation}  
\end{lemma}
According to lemma 1, the ESD of $B_n$ is determined by the ESD of $A_N$ and $T_N$. Then an ideal matrix could be constructed, which has the similar ESD as \eqref{th1LN} :

\begin{lemma}
When $P\rightarrow \infty, N\rightarrow \infty, P/N\rightarrow c$, for
\begin{equation}\label{RA}
\begin{split}
&\hat{\mathbf{R}}_x=(\hat{\mathbf{H}}+\hat{\mathbf{W}})(\hat{\mathbf{H}}+\hat{\mathbf{W}})^*/N\\
&\hat{\mathbf{H}}=\sqrt{N}\mathbf{H}\mathbf{E}_L\in \mathbb{C}^{P\times N},
\end{split}
\end{equation}
where $\mathbf{E}_L=(\mathbf{I}_L \quad \mathbf{0}_{L\times(N-L)})$ as a select matrix, $\hat{\mathbf{W}}=(\mathbf{Q}^* \ \mathbf{\tilde{W}}_{P\times ({N-L})})$, where $\mathbf{\tilde{W}}_{P\times ({N-L})}$ is the last $N-L$ columns of $\mathbf{W}$. Then the ESD of \eqref{RA} and \eqref{th1LN} converge to the same PDF.
\end{lemma}

\begin{proof}
 Considering \eqref{RA}, we have
\begin{equation}\label{mspike}
\begin{split}
&\hat{\mathbf{R}}_x=\hat{\mathbf{T}}+\hat{\mathbf{W}}\hat{\mathbf{W}}^*/N
\end{split}
\end{equation}
where
\begin{equation}\label{th2LN}
\begin{split}
&\hat{\mathbf{T}}=N\mathbf{H}\mathbf{H}^*/N+(\mathbf{HE}_L\hat{\mathbf{W}}^*+\hat{\mathbf{W}}\mathbf{E}_L^*\mathbf{H}^*)/\sqrt{N}\\
&=\mathbf{H}\mathbf{H}^*+(\mathbf{H}\mathbf{Q}+\mathbf{Q}^*\mathbf{H}^*)/\sqrt{N}.
\end{split}
\end{equation} 

Thus the entries of $\mathbf{T}$ and $\hat{\mathbf{T}}$ is the same. Obviously the ESD of $\mathbf{T}$ and $\hat{\mathbf{T}}$ is the same, so the eigenvalue of $\mathbf{T}$ and $\hat{\mathbf{T}}$ is the same. On the other hand, according to Theorem 1, the ESD of $\mathbf{W}\mathbf{W}^*/N$ and $\hat{\mathbf{W}}\hat{\mathbf{W}}^*/N$ converges to M-P law. Using Lemma 1, entries in $\mathbf{R}_x$ and $\mathbf{\hat{R}}_x$ converges to the same PDF. 
\end{proof}

In lemma 2, an artificial $\mathbf{\hat{R}}_x$ is constructed whose ESD converges to the same distribution as $\mathbf{R}_x$. Though $\mathbf{\hat{R}}_x$ is an artificial matrix which is impossible to obtain, it could be substituted by central spiked Wishart matrices based on convergence in distribution, which is:

\begin{theorem}
The ESD of SCM \eqref{suma} and a central Wishart matrix $\mathbf{R_z}$ converge to the same PDF, where
\begin{equation}\label{wishart}
\mathbf{R_z}=\frac{1}{N}\mathbf{Z}\mathbf{\Sigma_N}\mathbf{Z}^*,
\end{equation}
 and  $\mathbf{\Sigma_N}=\rm{diag}(\underbrace{N\sigma^2+1,\cdots,N\sigma^2+1}_L,\underbrace{1,\cdots,1}_{N-L})$, $\mathbf{Z}(i,j)\sim \mathcal{CN}(0, 1)$, $i\in [0,P-1]$, $j\in [0,N-1]$, $P/N\rightarrow c$, $\sigma^2 \in O(1/N)$.
\end{theorem}

\begin{proof}
Let $\hat{\mathbf{X}}=\mathbf{\hat{H}}+\hat{\mathbf{W}}$. The elements of the first $L$ columns are $(\hat{\mathbf{H}}(i,j)+\hat{\mathbf{W}}(i,j))/\sqrt{N}$ which follows $\mathcal{CN}(0,(N\sigma^2+1)/N)$.  
 
For random matrix $\mathbf{R}_z=\frac{1}{\sqrt{N}}\mathbf{Z}{\mathbf{\Sigma}_N}^{1/2}(\frac{1}{\sqrt{N}}\mathbf{Z}{\mathbf{\Sigma}_N}^{1/2})^*$, the entires of first $L$ columns of $\frac{1}{\sqrt{N}}\mathbf{Z}{\mathbf{\Sigma}_N}^{1/2}$ follow the distribution $\mathcal{CN}(0,(N\sigma^2+1)/N)$, which is the same as $\hat{\mathbf{X}}$.
 
The entries of $\hat{\mathbf{X}}/\sqrt{N}$ and $\mathbf{Z}{\mathbf{\Sigma}_N}^{1/2}/\sqrt{N}$ have the same distribution, and the ESD of both matrices fits theorem 2 of \cite{silver}. So the ESD of $\hat{\mathbf{X}}\hat{\mathbf{X}}^*/N$ and $\mathbf{R}_z$ converge to the same PDF. According to Lemma 2, the ESD of $\mathbf{R}_x$ and $\mathbf{R}_z$ converge to the same PDF.
 \end{proof}

$\mathbf{R_z}$ is a central Wishart matrix, with the covariance matrix $\mathbf{\Sigma_N}$. This spike Wishart matrix has been widely studied\cite{wang14}, which could be used to design hypothesis detector in wireless communication.

\section{Application: Function of detection}

According to Lemma 1, the signal detection in large array is translated to a standard high dimensional signal detection problem. Classical methods include detecting the ratio of biggest and smallest eigenvalue, detecting the trace of covariance matrix\cite{lin12}, and LRT\cite{Debbah}. Among them, LRT has a decent history and has been widely used. 
 
According to chapter $10$ of \cite{Anderson}, we assume that $\mathbf{R}_n$ is the SCM formatted by \eqref{suma}, to test the hypothesis $\mathcal{H}_0: \mathbf{\Sigma}=\mathbf{I}$, where $\mathbf{\Sigma}$ is the covariance matrix of a vector $\mathbf{X}$ distributed according to $N(\mu,\Sigma)$. It is showed that the LRT
 \begin{equation}\label{GLRTL}
 D_n=\rm{tr}{{\bf{R}}_n}-{\rm{log}}({\rm{det}}{\mathbf{R}}_n)-P
 \end{equation}
 is unbiased LRT of \eqref{disc}. The detector is researched widely such as \cite{wang14}.

To make statistical integration about a parameter $\theta = \int f(x)dF^M(x)$. It is natural to use the estimator
\begin{equation}\label{GLRTL1}
\hat\theta_n = \int f_n(x)dF_n^M(x) = 1/P\sum_{i=1}^Pf_n(\lambda_i^M)
\end{equation}

The log-liklihood ratio (LLR) could be further rewritten as
\begin{equation}\label{LRTS}
D_n = \sum_{i=0}^{P-1}\lambda_{n,i}-{\rm{log}}\prod_{i=1}^P\lambda_{n,i}-P=\sum_{i=1}^P (\lambda_{n,i}-{\rm{log}}\lambda_{n,i}-1)
\end{equation}
where $\lambda_i$ is the $i$-th eigenvalue of tested covariance matrix.

The ESD of $\mathbf{R}_z$ has been rarely studied. The work \cite{wang14} investigated the fluctuation of linear spectral statistics of form \eqref{GLRTL}, with the form $\mathbf{R}_n=\mathbf{Z}^*\mathbf{\Sigma}_P\mathbf{Z}/N$ or $\mathbf{\underline{R}}_n=\mathbf{\Sigma}_P^{1/2}\mathbf{Z}\mathbf{Z}^*\mathbf{\Sigma}_P^{1/2}/N$ and its results are based on $P/N<1$, which is different from $\mathbf{R}_z$.

Under $P/N\geq1$, it has
 \begin{equation}\label{spike4}
\mathbf{R_z}=\frac{1}{N}\mathbf{Z}\mathbf{\Sigma_N}\mathbf{Z}^*=\frac{1}{N}(\mathbf{Z^*})^*\mathbf{\Sigma_N}\mathbf{Z}^*
\end{equation}
Substituting $\mathbf{Z}$ by $\mathbf{Z}^*$, $\mathbf{R_z}$ could fit the structure of $\mathbf{R}_n$ in \cite{wang14}, and the results from \cite{wang14} could be used directly. But in practice, the number of samples is larger than the antenna number in most scenario, which means $0<c<1$ is more practical.  

Under $0<c<1$ and $\mathcal{H}_0$, $\mathbf{R}_z$ equals $\mathbf{R_{w,w}}$ and $D_n$ converges to normal distribution as shown in \cite{wang14}, which is
\begin{equation}
D\rightarrow N(\mu_{D,H_0},\sigma_{D,H_0}^2)
\end{equation}
where
\begin{equation}\label{h0para}
\begin{split}
&\mu_{D,{H_0}}=P(1-\frac{c-1}{c}{\rm{log}}(1-c))-\frac{{\rm{log}}(1-c)}{2}\\
&\sigma_{D,H_0}^2=-2{\rm{log}}(1-c)-2c
\end{split}
\end{equation}

Under $0<c<1$ and $\mathcal{H}_1$, it has

$Theorem \:2$: The ESD of
$\mathbf{R}_z=\mathbf{Z\Sigma_N Z}^*/N$ and $\tilde{\mathbf{R}}_z=\frac{1}{N}\mathbf{\Sigma_P}^{1/2}\mathbf{Z}\mathbf{Z}^*\mathbf{\Sigma_P}^{1/2}$converge to the same PDF, where $\mathbf{\Sigma_P}=\rm{diag}(\underbrace{P\sigma^2+1,P\sigma^2+1,\cdots,P\sigma^2+1}_L, \underbrace{1,\cdots,1}_{P-L})$.

\begin{proof}
According to classical matrix theory, the non-zero eigenvalues of $\mathbf{R}_z$ and $\mathbf{\hat{R}_z}= \frac{1}{N}\mathbf{\Sigma_N}^{1/2}\mathbf{Z}^*\mathbf{Z}\mathbf{\Sigma_N}^{1/2}$ are the same.

As introduced in \cite{baik}, the outlier eigenvalues of $\hat{\mathbf{R}}_z$ and $\tilde{\mathbf{R}}_z$ could be elaborated by $\mathbf{\Sigma}_P$ and $\mathbf{\Sigma}_N$.

For $\tilde{\mathbf{R}}_z$, according to Theorem 1.1 of \cite{baik}, it has
\begin{equation}\label{eofR}
\lambda_{\tilde{\mathbf{R}}_z,l}\rightarrow\lambda_{\mathbf{\Sigma}_P,l}+\cfrac{c\lambda_{\mathbf{\Sigma}_P,l}}{\lambda_{\mathbf{\Sigma}_P,l}-1}, 0\leq l\leq L-1
\end{equation}
For $\hat{\mathbf{R}}_z$, using the same method as in \eqref{spike4}, it is further written as
\begin{equation}\label{hatR}
\hat{\mathbf{R}}_z=c\times\frac{1}{P}\mathbf{\Sigma_N}^{1/2}\mathbf{Z^*}(\mathbf{Z}^*)^*\mathbf{\Sigma_N}^{1/2}
\end{equation}
According to Theorem 1.2 in \cite{baik}, the first $L$ outlier eigenvalues hold
\begin{equation}\label{eoftR}
\lambda_{\hat{\mathbf{R}}_z,l}\rightarrow\lambda_{\mathbf{\Sigma}_N,l}\times c+\cfrac{ \lambda_{\mathbf{\Sigma}_N,l}}{\lambda_{\mathbf\Sigma_N,l}-1}, 0\leq l\leq L-1
\end{equation}
	
According to the definition of $\Sigma_P$ and $\Sigma_N$, we have
\begin{equation}\label{T}
\lambda_{\mathbf{\Sigma}_P,l}-1=(\lambda_{\mathbf{\Sigma}_N,l}-1)\times c
\end{equation}
By calculation, it is easy to find that $\lambda_{\tilde{\mathbf{R}}_z,l}$ and $\lambda_{\hat{\mathbf{R}}_z}$ converge to same limitation when $N\rightarrow \infty$, which is $\lambda_{\mathbf{\Sigma}_P,l}+\cfrac{c\lambda_{\mathbf{\Sigma}_P,l}}{\lambda_{\mathbf{\Sigma}_P,l}-1}$. The other eigenvalues of both matrices converge to Machenko-Pastur distribution, as in \cite{baik} and \cite{wang14}. Thus, the ESD of $\hat{\mathbf{R}}_x$ and $\mathbf{\tilde{R}_z}$ converge to the same distribution asymptotically.% because [] and [] only needs the entries converge to particular distribution.
 \end{proof}

\begin{figure}\label{MPlawdensity}\centering
	\includegraphics[scale=0.5]{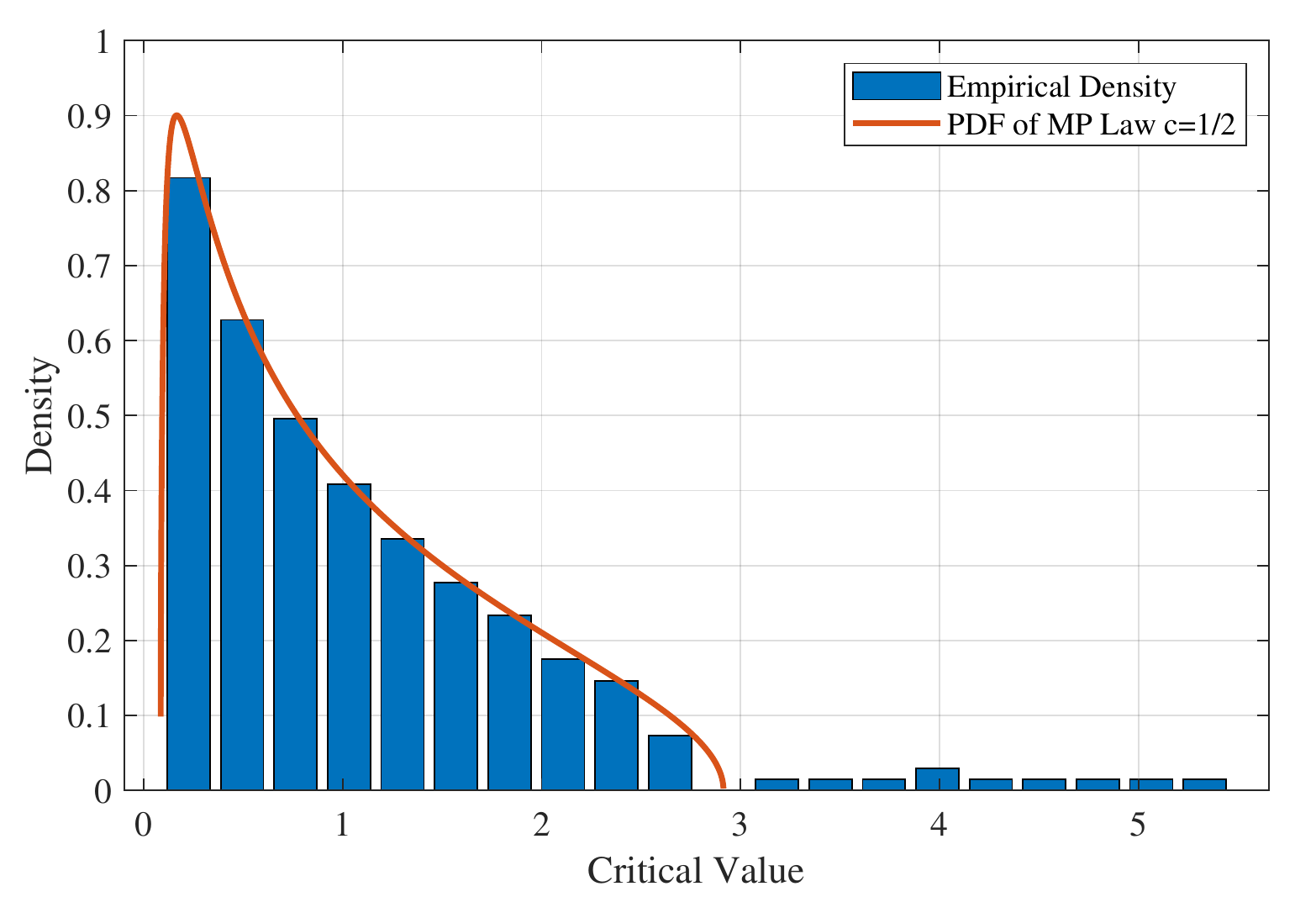}\\
	\caption{The eigenvalues density of $\mathbf{R}_x$ when $\mathcal{H}_1, P=256, N=1024, L=10, \sigma ^2=1/100.$}
\end{figure}

Fig. 2 demonstrates the ESD density of $\mathbf{R}_x$. The largest $L$ eigenvalues are larger than $(1+\sqrt{c})^2$. The rest follow M-P Law, which coincides with our analysis. This distribution is quite different with the eigenvalues of covariance matrix $(\mathbf{T}_P+\mathbf{W}\mathbf{W}^*)$, which would stacked at $1$ and $3.56$.

Combining Theorem 1 and Theorem 2, we have

$Corollary$:
The ESD of SCM \eqref{suma} and $\tilde{\mathbf{R}}_z=\frac{1}{N}\mathbf{\Sigma_P}^{1/2}\mathbf{Z}\mathbf{Z}^*\mathbf{\Sigma_P}^{1/2}$ converge to the same distribution.
	
Based on \cite{wang14}, the LLR $D$ in \eqref{LRTS} converges asymptotically to a Normal distribution $\mathcal{N}(\mu_{D,H_1},\sigma^2_{D,H_1})$ with
\begin{equation}\label{h1mean}
\begin{split}
\mu_{D,H_1}&=P\times(1+\frac{1}{P}L(P\sigma^2+1)-\frac{L}{P}-\\
&\frac{1}{P} L \rm{log}( P\sigma^2+1)-(1-\frac{1}{c})\rm{log}(1-c))+log(1-c)/2\\
\sigma^2=&{-2\rm{log}(1-c)-2c}
\end{split}
\end{equation}

Fig. 3 shows the distribution of $D_n$ under $\mathbf{H}_0$ and $\mathbf{H}_1$, with $P=256, N=512, L=10, SNR=-15.5 dB$. The simulated distributions look in good agreement with the theoretical distributions.

A classical signal detection algorithm is used (algorithm 1 in \cite{lin12}) where we use the metric \eqref{LRTS} as the core of algorithm. 

\begin{center} 
\begin{tabular}{llrr}
	\toprule
	Algorithm1: Signal detection using GLRT \\
	\midrule
	1: At receiver, gets the received data \\from each port synchronously and standardized respectively.\\
	2: Organize the received signals together and form a matrix\\ as \eqref{dh}
	and calculate the sample covariance matrix. \\
	3: Compute $G'$ using the eigenvalues of\\ the sample covariance matrix by \eqref{LRTS}.\\
	4: Determine the probability of false alarm $P_{fa}$ and\\
	find out the threshold $\gamma$ via computations.\\
	5: if $G > B_L$ then\\
	6: signal exists;\\
	7: else\\
	8: signal does not exist.\\
	9: end if\\
	\bottomrule
\end{tabular}
\end{center}

\begin{figure}\label{density}\centering
	\includegraphics[scale=0.4]{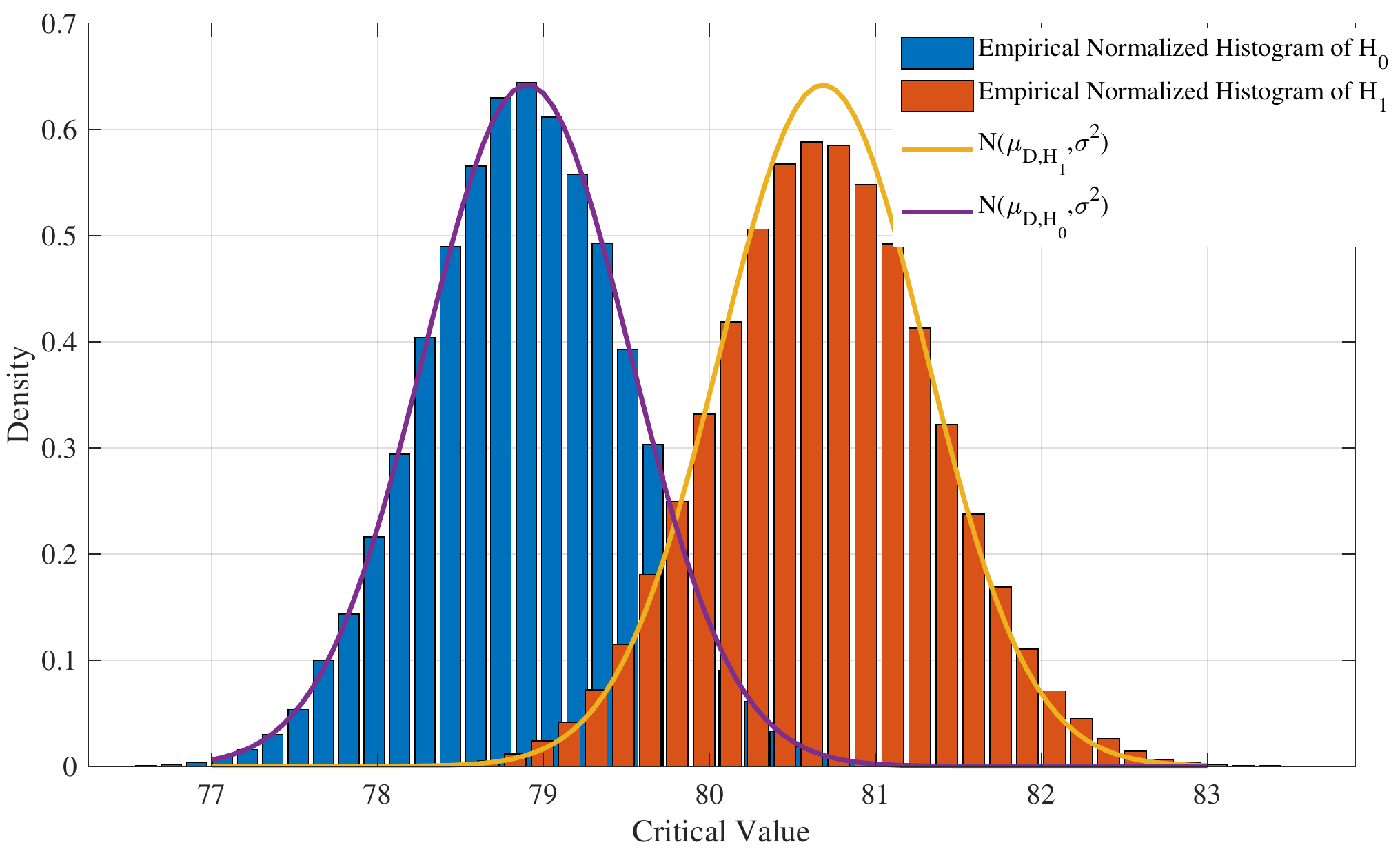}\\
	\caption{The $D_n$ (17) distributions with $\mathcal{H}_1: P=256, N=512, L=10, SNR = -15.5 dB, \mathcal{H}_0: P=256, N=1024$, simulations and theoretical, respectively.}
\end{figure}

As in literature \cite{wang14}, when $P\rightarrow \infty$ and $N\rightarrow\infty$, the variance is the same as that of $\mathcal{H}_0$.

Interestingly, from \eqref{h1mean}, we find that the mean of $D$ is only related with $L$ and $P$, and the variance is only related with $c$ no matter signal present or not.

The choice of threshold $\gamma$ is a compromise between $P_d$ and
$P_{fa}$. The probability of false alarm is
\begin{equation}
\begin{split}
P_{fa}& = P(D >\gamma|H_0)\\
&= P(\frac{D-\mu_{D_0,H_0}}{\sigma}>\gamma)\\
&=Q(\gamma)
\end{split}
\end{equation}
where the Q function is one minus the cumulative distribution function of the standardized normal random variable.

Under the limitation of the false alarming probability $P_{fa}$, the false alarm boundary is set as
\begin{equation}
\gamma=Q^{-1}(P_{fa})
\end{equation}

Under $H_1$, using $\gamma$ as threshold, the distance between $D_n$ and $\gamma$ follows Normal distribution as
\begin{equation}\label{all}
\begin{split}
&G^{'}\rightarrow N(\frac{\mu_{D,H_1}-\mu_{D,H_0}}{\sigma}-\gamma,1)
\end{split}
\end{equation}the theoretical miss probability is 
\begin{equation}\label{misspro}
P_{la}=1-P_d
=Q(\frac{\mu_{D,H_1}-\mu_{D,H_0}}{\sigma}-\gamma)
\end{equation}

\section{Simulations}
\begin{figure}\label{meana}\centering
	\includegraphics[scale=0.45]{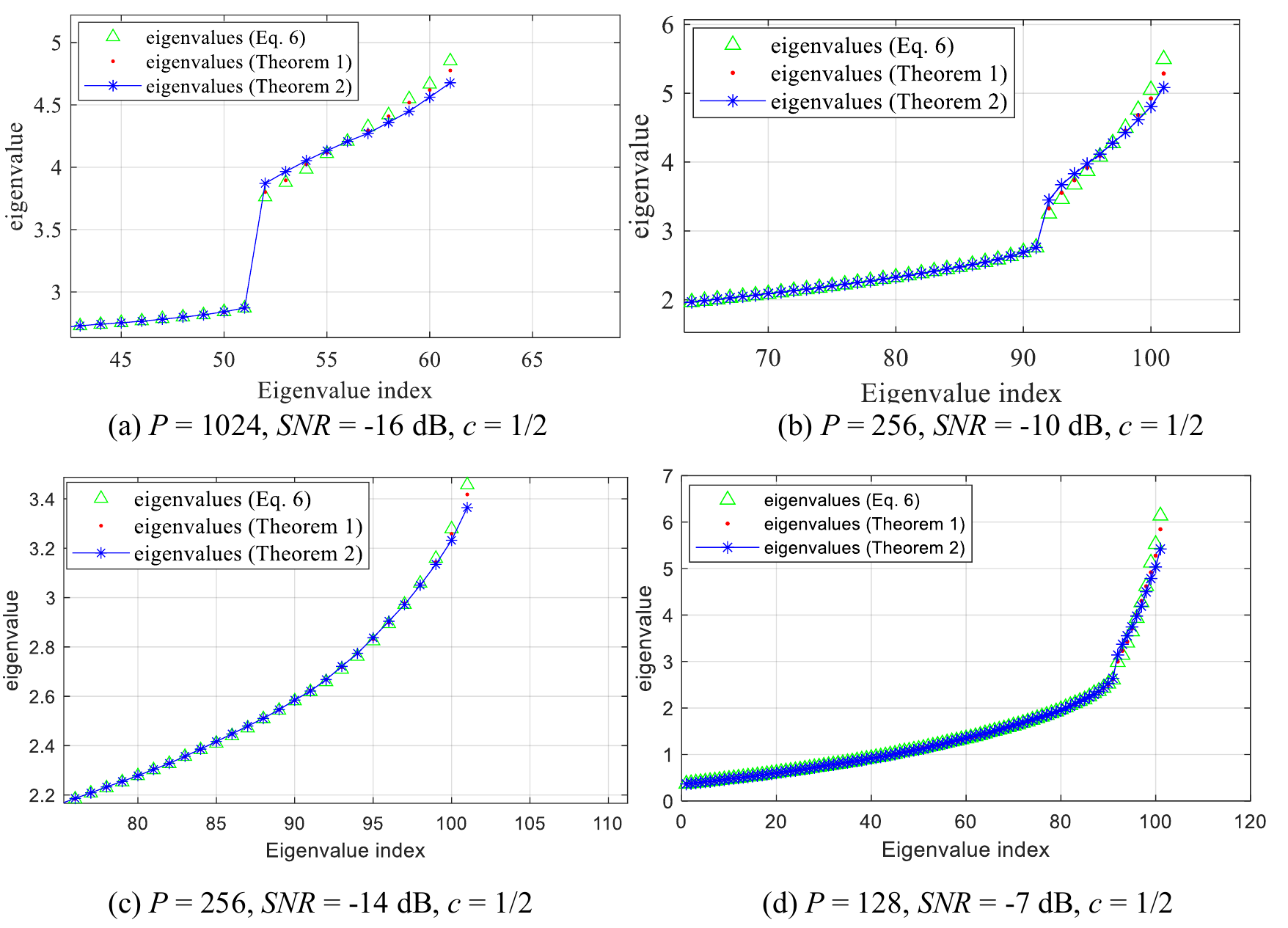}\\
	\caption{The eigenvalues of $\mathbf{R}_x$, $\mathbf{R}_z$ in Th. 1 and $\hat{\mathbf{R}}_z$ in Th. 2  when $\mathcal{H}_1$.}
\end{figure}

 The first simulation is to test the theoretical results about \eqref{suma}, (16) and Theorem 2. A channel with 10 taps is used. We choose $P=128, 256, 1024$, $c=1/2$, $SNR=-7, -10, -16$ dB, respectively. Then from Fig. 4, it is obvious that the eigenvalues of \eqref{suma}, \eqref{mspike} and Theorem 2 coincide with each other, and the ESD of the three matrices are the same. In Fig. 4, the gaps between the largest eigenvalues of the three matrices are smaller than $5\%$ of the value. With the increment of $P$ and decrement of SNR, these gaps are smaller.  

The second part of simulations is to detect LLR $G'$. In Fig. 5, we choose $P=256$ and $N=512, 1024, 2048$, respectively. The $P_{fa}$ is set as $0.05$. Using \eqref{all}, the theoretical miss probability is obtained. The algorithm aided by signal detection simulations is based on Algorithm 1 and \eqref{LRTS}. It is shown that, with the increment of sample number, the detection performance of the algorithm is improved. With the sample doubled, the detection threshold would reduce $2$ dB. In simulations, though the results of theoretical derivations and simulations are very close, the theoretical results look over-optimized the detection by smaller than 0.4 $dB$. Fig. 6 is the ROC (receiver operating characteristic) curve with $L=1, P=256, c=1/2, SNR = -16$ dB and $-15.5$ dB. In Fig. 6, the over-optimization is also found, where the ROC in theoretical is slightly better than empirical.
	
The gaps in Fig. 4, 5 and 6 are brought by the infinite approximation of $P$, as well as the approximation character of variance estimation which is referred from \cite{wang14}. Reviewing Fig. 3, the variance under $\mathcal{H}_1$ looks larger than under $\mathcal{H}_0$, which makes the empirical density of $\mathcal{H}_0$ higher than $\mathcal{H}_1$ but the theoretical variance is still the same. So even the variance analysis of (24) in \cite{wang14} is the most accurate work till now, there still has room to improve. 

\begin{figure}\centering
	\includegraphics[scale=0.4]{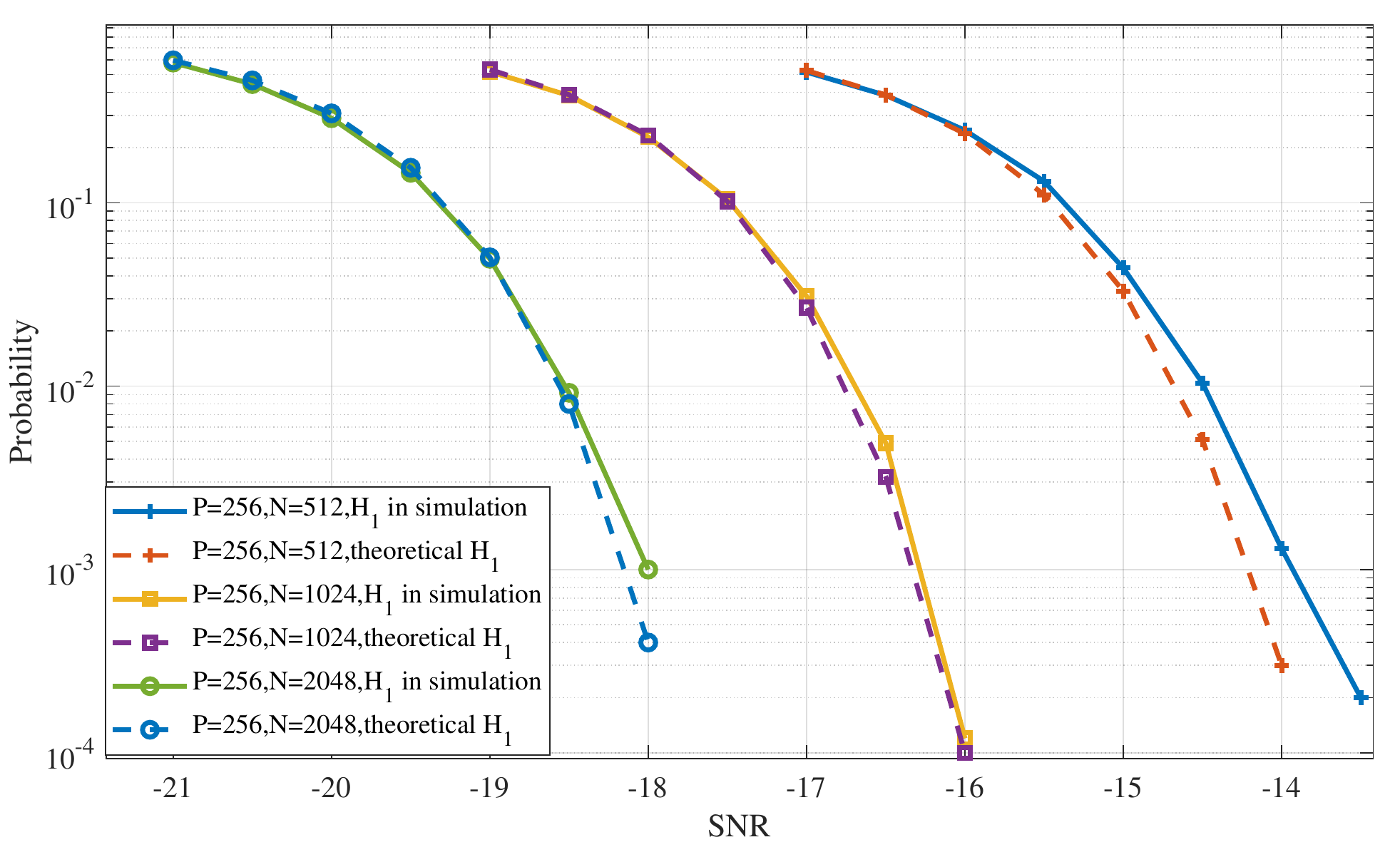}\\
	\caption{Miss probability under different antenna number with different SNR}
\end{figure}

\begin{figure}\centering
	\includegraphics[scale=0.5]{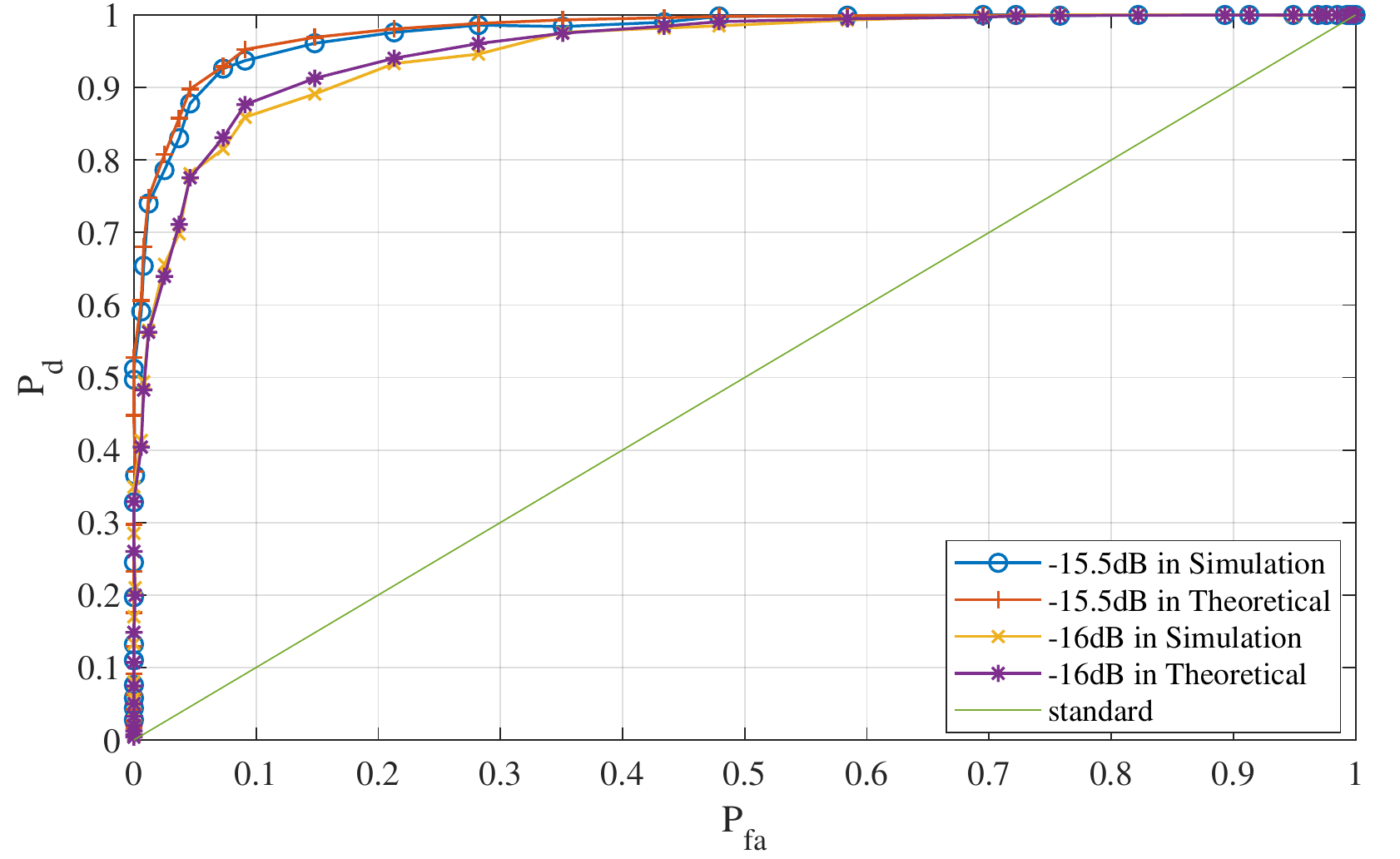}\\
	\caption{ROC curve under $L=1, P=256, c=1/2, SNR = -16dB$ and $-15.5 dB$.}
\end{figure}

\section{Conclusion}

In this paper, we have proposed a detection algorithm in the large antenna array system when the antenna number and the samples number are comparable. Two theorems have been presented to connect the detection with the newest statistical results. Furthermore, we have given a new detection metric which is based on an asymptotic result in RMT. Simulation results have shown that the model agrees with real simulations very well. In further works, we would simplify the complexity using distributed algorithm.
This paper has given results in asymptotic scenario. It is interesting to investigate the bound when $P$ and $N$ are not so large. Furthermore, this results could be extended to other applications easily, such as a multiuser with Rician channel scenario.

\end{document}